\documentclass[aps,prx,onecolumn,superscriptaddress,notitlepage]{revtex4-1}

\usepackage[T1]{fontenc}
\usepackage[utf8]{inputenc}
\usepackage{lmodern}
\usepackage{amsmath}
\usepackage{amssymb}
\usepackage{amsthm}
\usepackage{amsbsy}
\usepackage{mathrsfs}
\usepackage{graphicx}
\usepackage{enumerate}
\usepackage{xcolor}   
\usepackage{array}
\usepackage{enumitem}
\usepackage{tikz}
\usepackage{accents}
\usepackage{dsfont}
\usepackage{hyperref}

\newcommand{\accentfuturearrowc}[1]{%
\begin{tikzpicture}[#1]%
\fill (2mm,0) -- (1.2mm,0.3mm) -- (1.2mm,-0.3mm);
\draw[line width = 0.2mm] (-0.6mm,0mm) -- (1.2mm,0mm);
\draw (2mm,0) -- (1.2mm,0.3mm) -- (1.2mm,-0.3mm) -- cycle;
\end{tikzpicture}%
}
\newcommand{\tempfutc}{\accentfuturearrowc{}}

\newcommand{\futp}[1]{\accentset{\tempfutc}{#1}}

\newcommand{\temppast}{\scalebox{-1}[1]{\accentfuturearrowc{}}}
\newcommand{\past}[1]{\accentset{\temppast}{#1}}

\newcommand{\accentbotharrowc}[1]{%
\begin{tikzpicture}[#1]%
\fill (-0.6mm,0) -- (0.2mm,0.3mm) -- (0.2mm,-0.3mm);
\fill (2mm,0) -- (1.2mm,0.3mm) -- (1.2mm,-0.3mm);
\draw[line width = 0.2mm] (-0.1mm,0mm) -- (1.2mm,0mm);
\draw (2mm,0) -- (1.2mm,0.3mm) -- (1.2mm,-0.3mm) -- cycle;
\draw (-0.6mm,0) -- (0.2mm,0.3mm) -- (0.2mm,-0.3mm) -- cycle;
\end{tikzpicture}%
}

\newcommand{\tempbothc}{\accentbotharrowc{}}

\newcommand{\bothp}[1]{\accentset{\tempbothc}{#1}}

\newcommand{\ket}[1]{\ensuremath{|{#1}\rangle}}

\newcommand{\proj}[2]{\ensuremath{|{#1}\rangle\langle{#2}|}}

\newcommand{\Rc}[1]{\futp{c}_{#1}}
\newcommand{\Lc}[1]{\past{c}_{#1}}
\newcommand{\LRc}{\bothp{c}}

\theoremstyle{plain}
\newtheorem{thm}{Theorem}

\usepackage{mfirstuc}

\begin{document}

\title{Allowing Wigner's friend to sequentially measure incompatible observables}
\author{An\'{i}bal Utreras-Alarc\'{o}n}
\affiliation{Centre for Quantum Computation and Communication Technology (Australian Research Council Centre of Excellence), Centre for Quantum Dynamics, Griffith University, Brisbane, QLD 4111, Australia.}
\author{Eric G. Cavalcanti}
\affiliation{Centre for Quantum Dynamics, Griffith University, Gold Coast, QLD 4222, Australia.}
\author{Howard M. Wiseman}
\affiliation{Centre for Quantum Computation and Communication Technology (Australian Research Council Centre of Excellence), Centre for Quantum Dynamics, Griffith University, Brisbane, QLD 4111, Australia.}


\begin{abstract}
The Wigner's friend thought experiment has gained a resurgence of interest in recent years thanks to no-go theorems that extend it to Bell-like scenarios. One of these, by us and co-workers, showcased the contradiction that arises between quantum theory and a set of assumptions, weaker than those in Bell’s theorem, which we named ``local friendliness''. Using these assumptions it is possible to arrive at a set of inequalities for a given scenario, and, in general, some of these inequalities will be harder to violate than the Bell inequalities for the same scenario. A crucial feature of the extended Wigner’s friend scenario in our aforementioned work was the ability of a superobserver to reverse the unitary evolution that gives rise to their friend's measurement. Here, we present a new scenario where the superobserver can interact with the friend repeatedly in a single experimental instance, either by asking them directly for their result, thus ending that instance, or by reversing their measurement and instructing them to perform a new one. We show that, in these scenarios, the local friendliness inequalities will always be the same as Bell inequalities.
\end{abstract}

\maketitle

\section{Introduction}
\label{ch:introduction}

The Wigner's friend thought experiment highlights the discrepancy between the two manners in which a system evolves, according to quantum theory: either in a deterministic, reversible fashion for closed systems, or following the typically indeterministic and irreversible state-update rule, after a measurement. In Wigner's original scenario \cite{Wigner61}, two observers are considered: one, whom we will refer to as the ``friend'', measures a quantum system, describing the resulting state acording to the state-update rule. The other, known as the ``superoberserver'', is outside the closed laboratory the friend performs their measurement in, and can describe the evolution of the contents of this laboratory using a unitary operator. This includes the measurement apparatus interacting with the quantum system, the friend recording the observed outcome, etc. The superobserver thus describes the friend as being in a superposition of observing all results.

In recent years, generalizations of the Wigner's friend experiment, incorporating elements of a Bell-like scenario, have been proposed in order to harden the apparent discrepancy above into novel no-go theorems. The first of these was due to Brukner \cite{Brukner18}. He considered an assumption he called ``observer-independent facts'', along with the previously established assumptions such as locality and freedom of choice, although not predetermination \cite{Brukner18}. Other works have further developed the ideas proposed in \cite{Brukner18}, with analyses considering different assumptions and/or scenarios \cite{Frauchiger18,Guerin21,Wiseman22,Xu23}.

Most relevant to this article is \cite{Bong20}, where we, along with other authors, coined the term ``local friendliness'' (LF) to describe a set of assumptions weaker than those considered by Brukner in \cite{Brukner18}. We then proved, using an extended Wigner's friend scenario similar to the one presented in \cite{Brukner18}, that LF is incompatible with quantum theory. There was also a difference in our presentation of the scenario in \cite{Bong20} and that originally employed by Brukner: the sole ``superpower'' we give to the superobserver is the ability to reverse the measurement performed by the friend, along with his memory and any other record of it. The result of \cite{Bong20} is akin to Bell's theorem \cite{Bell64}, albeit stronger due to no longer considering the assumption of predetermination. In \cite{Bong20} we also showed that the LF inequalities (those whose violation would serve as a proof for the LF no-go theorem) are, in general, different from the previously known Bell inequalities.

The above work, especially the LF no-go theorem, raises interesting new questions and directions. Is the gap between Bell inequality violation and LF inequality violation a fundamental property of the LF assumptions themselves, or is it a limitation of the scenarios to which they are being applied? In support of the latter possibility, it is worth noting the following. In \cite{Bong20} we observed that, for a scenario with two dichotomic measurements per party, the LF inequalities reduce to the Bell inequalities for that scenario. The reversal of the friend's state is another interesting feature. What others scenarios can be constructed that include such reversal, and what can they be used for? The present work addresses both of these questions.

Here, we consider an extended Wigner's friend scenario where a superobserver can, at multiple points in time, choose to open the friend's laboratory, thus ending the experiment at that point, or to reverse the evolution of the friend and instruct them to perform a new measurement. After the final such time, if the superobserver has not yet asked the friend for his result, she may make her own measurement, in a fixed basis. We refer to this scenario as a \emph{sequential} extendend Wigner's friend scenario. Under the LF assumptions, the correlations that arise in such scenario are indeed constrained solely by Bell inequalities, independently of the number of measurements under consideration, for a suitable sequence of reversals and instructions. 

This paper is structured as follows. In Section \ref{ch:definitions} we define some basic terms and mathematical concepts necessary for what follows. Section \ref{ch:background} briefly summarizes the Wigner's friend experiment, Bell's theorem, Brukner's theorem and the LF theorem. Section \ref{ch:sewfs} describes the sequential extended Wigner's friend scenario, setting up the paper's main result, which is presented and proven in Section \ref{ch:theorem}. Finally, in Section \ref{ch:discussion} we discuss our results.

\section{Definitions}
\label{ch:definitions}

\subsection{Basic postulates and concepts}
Through the entirety this paper, we will make use of the following concepts:
\begin{itemize}
    \item \textbf{Absoluteness of Observed Events (AOE):} A space-time variable observed by any observer takes an absolute singular value, and is not ``relative'' to anything or anyone.
    
    \item \textbf{Local agency:} The only relevant events correlated with a free choice are in its future light cone.
    
    \item \textbf{Predetermination:} There exists a foliation $\mathfrak{F}$ such that any observable space-time variable $A$ is determined by a sufficient specification of space-time variables on any space-like hypersurface $\mathfrak{S}\in\mathfrak{F}$ prior to $A$, possibly in conjunction with free choices subsequent to $\mathfrak{S}$.
\end{itemize}

We define \textbf{local determinism} as the conjunction of AOE, local agency and predetermination. On the other hand, \textbf{local friendliness} is defined as the conjunction of only AOE and local agency.

We will also make use of the following, empirically well supported, principle:

\begin{itemize}
    \item \textbf{No-signalling principle:} The empirical probabilities associated to any set of observable space-time variables are unchanged by conditioning on any free choice space-like separated from every member of the set.
\end{itemize}

\subsection{Behaviours and polytopes of relevance}

Let us consider two distant experimenters, Alice and Bob, who share a bipartite physical system and perform a measurement on their respective subsystems.
Alice makes a choice between $M_A$ measurement settings, denoted by $x$, with $N_A$ possible outcomes, labelled by $a$. Likewise, Bob's input and output are represented by $y$ and $b$ from $M_B$ and $N_B$ possible values, respectively. Let $\mathcal{X}$ be the set of all possible inputs $x$, and $\mathcal{Y}$ the set of all possible $y$. We define $\mathcal{A}$ as the set of all possible outputs for Alice's measurements, and $\mathcal{B}$ as the set of outputs for Bob's measurements. Then, we can define a \emph{public scenario} as the tetrad $\mathcal{S = (A,B,X,Y)}$.

Given the public scenario $\mathcal{S}$, we can define the probabilities $p(ab|xy)$ for all inputs and outputs. We define the \emph{behaviour}, or \emph{correlations}, of the system in this scenario as $\bar{p}=(p(ab|xy):(a,b,x,y)\in A \times B \times \mathcal{X} \times \mathcal{Y})$, a point in $\mathbb{R}^{M_A M_B N_A N_B}$ with each coordinate corresponding to a probability for a particular pairing between joint inputs and outputs. Naturally, we require behaviours to satisfy normalization and positivity constraints:
\begin{equation}
\label{eq:normalization}
    \sum_{a,b} p(ab|xy) = 1 \;\;\; \forall x,y,
\end{equation}
\begin{equation}
\label{eq:postivity}
    p(ab|xy) \ge 0 \;\;\; \forall a,b,x,y.
\end{equation}

For a public scenario $\mathcal{S}$, we denote the set of all behaviours that satisfy the 
no-signalling principle as $\mathbb{NS}(\mathcal{S})$. For all public scenarios, $\mathbb{NS}(\mathcal{S})$ will be a \emph{convex polytope}, that is, a bounded convex set with flat sides. Behaviours in $\mathbb{NS}(\mathcal{S})$ have to satisfy the following conditions:
\begin{equation}
    p(a|xy) = p(a|x) \;\;\; \forall a,b,x,y,
\end{equation}
\begin{equation}
    p(b|xy) = p(b|y) \;\;\; \forall a,b,x,y.
\end{equation}

We denote the set of all behaviours that satisfy the conditions for local determinism within a public scenario $\mathcal{S}$ as $\mathbb{LD}(\mathcal{S})$. Like $\mathbb{NS}(\mathcal{S})$, $\mathbb{LD}(\mathcal{S})$ is a \emph{convex polytope}. Further characterization of $\mathbb{LD}(\mathcal{S})$ will be given in \ref{ch:background}\ref{sec:bell}.

We say that a correlation is deterministic for an input $i$ in $\mathcal{I}_X \subseteq \mathcal{X}$ if $p(a|x=i,y)\in\{0,1\}$ for all $y$, and similarly for inputs in $\mathcal{I}_Y \subseteq \mathcal{Y}$. Following Woodhead \cite{Woodhead14}, we define a \textbf{partially deterministic} set of behaviours $\mathbb{PD}_{\mathcal{I}_X,\mathcal{I}_Y}(\mathcal{S})$, with $\mathcal{I}_X\subseteq\mathcal{X}$ and $\mathcal{I}_Y\subseteq\mathcal{Y}$, to be the convex hull generated by points corresponding to correlations that are deterministic for inputs in $\mathcal{I}_X$ or $\mathcal{I}_Y$, but behave as the extreme points of $\mathbb{NS}(\mathcal{S})$ otherwise.

Regarding the relationships between these sets, we can see, for a given scenario $\mathcal{S}$, that, for any $\mathcal{I}_X$ and $\mathcal{I}_Y$,
\begin{equation}
    \mathbb{LD}(\mathcal{S}) \subseteq \mathbb{PD}_{\mathcal{I}_X,\mathcal{I}_Y}(\mathcal{S}) \subseteq \mathbb{NS}(\mathcal{S}).
\end{equation}
As particular cases, we have that $\mathbb{NS}(\mathcal{S}) = \mathbb{PD}_{\emptyset,\emptyset}(\mathcal{S})$, and $\mathbb{LD}(\mathcal{S}) = \mathbb{PD}_{\mathcal{X},\mathcal{Y}}(\mathcal{S})$. Also, it was shown in \cite{Woodhead14} that if, in a bipartite public scenario, all of Alice's inputs except $x=k$ are deterministic, then the respective partial deterministic polytope would be the same as the one as for local deterministic behaviours. That is to say:
\begin{equation}
\label{eq:pd:to:ld}
    \mathbb{PD}_{\mathcal{X}/k,\mathcal{I}_Y}(\mathcal{S}) = \mathbb{LD}(\mathcal{S}),
\end{equation}
for any set $\mathcal{I}_B$.



\section{Background}
\label{ch:background}

\subsection{The Wigner's friend thought experiment}

Consider a quantum system $S$ that can be in two possible orthogonal states, $\ket{\psi_1}_S$ or $\ket{\psi_2}_S$. Interacting with the system is an observer, the titular ``friend'', as introduced by Wigner \cite{Wigner61}. A second observer, or rather a \textit{superobserver}, describes the friend, along with his measurement apparatus and all other contents of his laboratory sans the system $S$, as another quantum system $F$. When interacting with the system, if the original state was $\ket{\psi_1}_S$ then the friend will see a flash of light and the state of the composite system will be $\ket{\varphi_1}_F\ket{\psi_1}_S$. If the system is originally in state $\ket{\psi_2}_S$ instead, the friend will not see a flash and the composite system will be in the state $\ket{\varphi_2}_F\ket{\psi_2}_S$. Here, the states $\ket{\varphi_1}_F$ and $\ket{\varphi_2}_F$ describe the state of the friend: when asked by the superobserver if he saw a flash he will say ``Yes'' if he is in state $\ket{\varphi_1}_F$, or ``No'' if he is state $\ket{\varphi_2}_F$.

Now, let us suppose the original state of system $S$ is some linear combination $\alpha\ket{\psi_1}_S + \beta\ket{\psi_2}_S$. After the friend interacts with the system, the composite state would be, according to standard quantum mechanics, $\alpha\ket{\varphi_1}_F\ket{\psi_1}_S + \beta\ket{\varphi_2}_F\ket{\psi_2}_S$. Here we ignore other systems that may correlate with the friend, as they made no difference to Wigner's argument; they would simply make the superposition ``bigger''. Then, if the superobserver asks the friend whether he has seen a flash or not, he will say ``Yes'' with probability $|\alpha|^2$ and ``No'' with probability $|\beta|^2$.

According to the friend, however, the composite system after the interaction would have the state $\ket{\varphi_1}_F\ket{\psi_1}_S$ or $\ket{\varphi_2}_F\ket{\psi_2}_S$, as corresponds to the friend's answer, and not the linear combination of these states. Wigner interpreted this situation as suggesting that the friend is governed by a different set of rules from other physical system, apparently due to his status as a conscious being.

\subsection{Bell's theorem}
\label{sec:bell}

Based on the scenario presented in \cite{EPR}, Bell considered a public scenario where Alice's and Bob's measurements are space-like separated from each other. This was studied under the assumptions of local determinism. Within Bell's scenario, the individual assumptions are equivalent to the following mathematical relations.
\begin{itemize}
    \item \textbf{AOE:} $\exists\, p(ab|xy),\;\forall a,b,x,y.$
    \item \textbf{Predetermination:} $\exists\,\lambda:p(ab|\lambda xy)\in\{0,1\},\;\forall a,b,x,y.$
    \item \textbf{Local agency:} $p(a|xy\lambda) = p(a|x\lambda)$, $p(b|xy\lambda) = p(b|y\lambda)$ and $p(\lambda|xy) = p(\lambda)$, $\forall a,b,x,y,\lambda.$
\end{itemize}
The joint probability $p(ab|xy)$ in the AOE assumption has to be consistent with the empirical frequency $f(ab|xy)$ that can be computed at a later time when Alice and Bob communicate with each other. To be more specific, AOE means that $p(ab|xy)$ has a value independently of Alice and Bob communicating their results.

Under these conditions it is possible to characterize $\mathbb{LD}(\mathbb{S})$ by describing the facets of the polytope as inequalities. Such inequalities can then be used to experimentally test if some physical phenomenon can be explained by a theory that subscribes to local determinism. The specific inequalities will depend on the public scenario $\mathcal{S}$. For example, for a scenario with a choice between two inputs per party, and with two possible outcomes per input, we can describe $\mathbb{LD}(\mathcal{S})$ using the CH inequalities \cite{Clauser74} (along with some trivial positivity constraints) of the form:
\begin{equation}
    \label{eq:ch}
    p_{++}(A_1B_1) +p_{++}(A_1B_2) + p_{++}(A_2B_1) - p_{++}(A_2B_2) - p_{+}(A_1) - p_{+}(B_1) \le 0,
\end{equation}
where $p_{++}(A_i B_j)=p(a=+1,b=+1|x=i,y=j)$.

We can easily check that quantum theory provides a violation of Eq.~(\ref{eq:ch}). Indeed, let us consider the two-photon state in the entangled polarization state $\ket{\Phi^+} = \dfrac{1}{\sqrt{2}}(\ket{HH} + \ket{VV})$. Say Alice and Bob perform projective POVMs of the form $\{ \proj{\phi}{\phi}, \proj{\phi+\frac{\pi}{2}}{\phi+\frac{\pi}{2}}\}$, where the first operator corresponds to the outcome $+1$ and the second to the outcome $-1$, and
\begin{equation}
    \ket{\phi} = \cos{\phi}\ket{H} + \sin{\phi}\ket{V}.
\end{equation}
Then, if Alice's measurements correspond to $\phi=0$ for $x=1$ and $\phi=\frac{\pi}{4}$ for $x=2$, and if Bob's measurements correspond to $\phi=\frac{\pi}{8}$ for $y=1$ and $\phi=-\frac{\pi}{8}$ for $y=2$, it follows that 
\begin{equation}
    p_{++}(A_1B_1) +p_{++}(A_1B_2) + p_{++}(A_2B_1) - p_{++}(A_2B_2) - p_{+}(A_1) - p_{+}(B_1) = \dfrac{\sqrt{2}-1}{2}.
\end{equation}
This is a clear violation of Eq.~(\ref{eq:ch}), which servers as a proof for Bell's theorem:
\begin{thm}[Bell's theorem]
\label{thm:bell}
There exist quantum phenomena for which there is no model satisfying local determinism.
\end{thm}

\subsection{Brukner's theorem}

Let us first consider the specifics of a single Wigner's friend setup, as described by Brukner in \cite{Brukner18}. The friend is inside a closed laboratory and performs a measurement on a system $S$, for example, the polarization of a photon. We will consider a particle that is initially in the state $\ket{D}_S = \frac{1}{\sqrt{2}}(\ket{H}_S + \ket{V}_S)$, with the friend measuring in the projective basis $\{\proj{H}{H}_S, \proj{V}{V}_S\}$.


On the other hand, the superobserver is outside this laboratory and can perform quantum operations on the overall system composed of $S$ and all the other elements inside the lab (the measurement apparatus, the friend's memory and sensory organs, etc.), which we denote by $F$. The state of $F$ after the friend observes that the photon has horizontal or vertical polarization will be represented by $\ket{L_{H}}_F$ or $\ket{L_{V}}_F$, respectively. Therefore, the state of the composite system after the measurement will be
\begin{equation}
    \ket{\Phi^+}_{SF} = \dfrac{1}{\sqrt{2}}(\ket{H}_S \ket{L_{H}}_F + \ket{V}_S \ket{L_{V}}_F).
\end{equation}

We can then see, for example, that if the superobserver wanted to know what outcome was observed by his friend, this will correspond to performing the POVM $\{\mathds{1}_{S}\otimes\proj{L_H}{L_H}_F, \mathds{1}_{S}\otimes\proj{L_V}{L_V}_F\}$. Alternatively, he could perform a measurement that highlights the entanglement between $S$ and $F$, for example, the projective measurement in the basis formed by the states $\ket{\Phi^\pm}_{SF} = \frac{1}{\sqrt{2}}(\ket{H}_S \ket{L_{H}}_F \pm \ket{V}_S \ket{L_{V}}_F)$ and $\ket{\Psi^\pm}_{SF} = \frac{1}{\sqrt{2}}(\ket{H}_S \ket{L_{V}}_F \pm \ket{H}_S \ket{L_{V}}_F)$.

\begin{figure}[ht]
    \centering
    \includegraphics[width=0.8\linewidth]{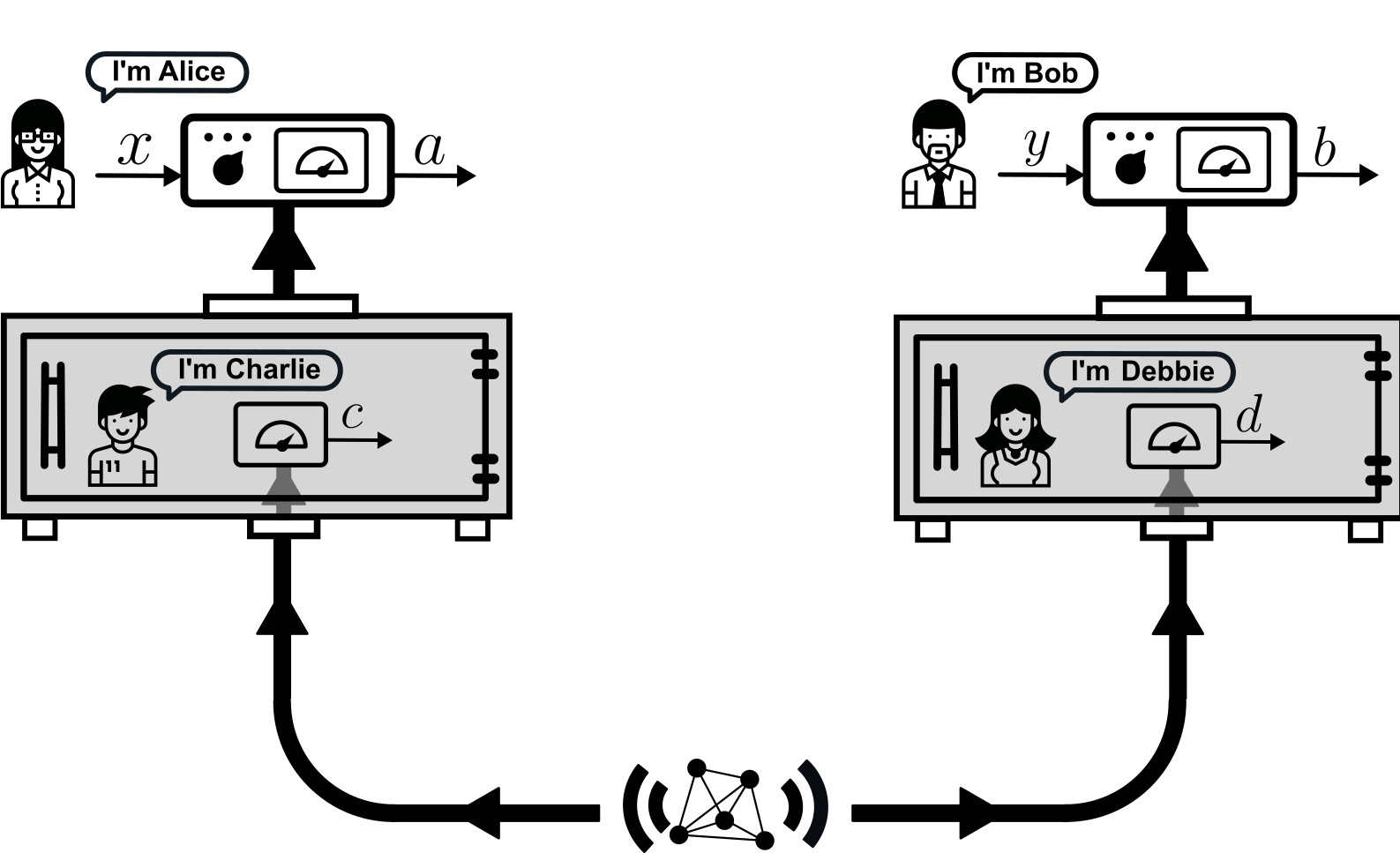}
    \caption{Extendend Wigner's friend scenario, originally presented by Brukner in \cite{Brukner18}. A pair of entangled particles is shared between Charlie and Debbie, two distant experimenters each inside a closed laboratory. Alice (Bob) is outside Charlie's (Debbie's) lab and can perform arbitrary quantum operations on it as a whole. Charlie (Debbie) performs a measurement on his (her) particle and observes the outcome $c$ ($d$), while Alice (Bob) describes this procedure as a unitary transformation. She (He) then performs a measurement, labelled $x$ ($y$), on the entire lab, with an outcome $a$ ($b$).}
    \label{fig:ewfs}
\end{figure}

Now, let us consider four experimenters, divided in two parts: Alice and Charlie in one region of space, and Bob and Debbie in another, as displayed in Fig.~\ref{fig:ewfs}. Alice and Charlie are in a ``Wigner's friend'' relationship, with Alice as a superobserver and Charlie as her friend. Likewise, Bob and Debbie are in the same relationship, with Bob as the superobserver and Debbie as the friend. Finally, Charlie and Debbie share a two-particle entangled system.

Brukner considers propositions about the outcomes of different observers, of the type (adapting to our terminology and notation): $A_1$: ``The pointer of Charlie's apparatus points to $H$'' and $A_2$: ``The pointer of Alice’s apparatus points to result $\Phi^+$''.  He then defines the postulate of ``observer-independent facts'', which requires ``an assignment of truth values to statements $A_1$ and $A_2$ independently of which measurement [Alice] performs''. Statements $B_1$ and $B_2$ for Bob/Debbie are defined analogously. Brukner's postulate reads:

\begin{quote}
     ``\textbf{Postulate 1}. \textit{(``Observer-independent facts'') The truth values of the propositions $A_i$ of all observers form a Boolean algebra $\mathscr{A}$. Moreover, the algebra is equipped with a (countably additive) positive measure $p(A)\ge 0$ for all statements $A\in \mathscr{A}$,  which is the probability for the statement to be true.''}
\end{quote}

Referring to the scenario in Fig.~\ref{fig:ewfs}, we set Charlie's outcome $c=1$ to correspond to statement $A_1$ being true, and $c=-1$ otherwise. Alice has a choice between either ($x=1$) asking Charlie for his outcome or ($x=2$) performing a projective measurement containing the projector onto $|\Phi^+\rangle$. Thus for $x=1$ Alice is inferring the value of $c$, and we set $a=1$ if she infers $c=1$, and $a=-1$ otherwise. For $x=2$ Alice's outcome is similarly set to $a=1$ if $A_2$ is true, $a=-1$ otherwise.

From this scenario, Brukner arrives at his theorem:
\begin{thm}[Brukner's theorem]
\label{thm:brk}
The following statements are incompatible:
\begin{enumerate}
    \item Universal validity of quantum theory: Quantum predictions hold at any scale, even if the measured system contains objects as large as an observer (including her laboratory, memory, etc.).
    \item Locality: The choice of the measurement settings of one observer has no influence on the outcomes of the other distant observer(s).
    \item Freedom of choice: The choice of measurement settings is statistically independent from the rest of the experiment.
    \item Observer-independent facts: One can jointly assign truth values to the propositions about observed outcomes (“facts”) of different observers (as specified in the postulate above).
\end{enumerate}
\end{thm}

Assumption (i) is what allows the superobservers to describe their (macroscopic) friends as quantum systems. According to Brukner, the conjunction of his assumptions (ii), (iii) and (iv) leads to a local deterministic model. This means that the possible correlations p(ab|xy) are constrained by Bell's inequalities. Then, in the example above, with 2 inputs per party and 2 outcomes per input, the correlations will be constrained by the CH inequalities, with their violation serving as a proof for Brukner's theorem.

However, as was discussed in \cite{Healey18, Bong20}, the way Brukner formalised his assumption (iv) (and in particular Postulate 1) implies the existence of a truth value for the four propositions $A_1,A_2,B_1,B_2$ regardless of which measurements are actually performed by Alice and Bob. In particular it implies the existence of a truth value for $A_2$, referring to Alice's measurement $x=2$, even when Alice does not perform that measurement (and similarly for $B_2$). In other words, the postulate of observer-independent facts implies the existence of a deterministic noncontextual model, which is in turn stronger than Bell's assumptions \cite{Mermin90}, and already ruled out by the Kochen-Specker theorem \cite{Kochen67}. This motivated our work in \cite{Bong20}, where we arrived at a stronger theorem than Brukner's by considering a weaker set of assumptions, namely, Local Friendliness.


\subsection{Local Friendliness no-go theorem}

The initial setup of the scenario in \cite{Bong20} is the same as shown in Fig.~\ref{fig:ewfs}, with Charlie and Debbie being inside their closed laboratories while sharing a two-particle system, and Alice and Bob being their superobservers, respectively. Charlie and Debbie make a measurement on their respective particle, and we label their corresponding outcomes as $c$ and $d$.

Alice's and Bob's measurements, however, will be different from those in \cite{Brukner18}. Here, each superobserver can choose between $M$ possible measurement settings, with $M\ge2$, whereas in \cite{Brukner18} they only had 2 choices. This generalisation to $M>2$ allows the demonstration of interesting new phenomena, as will be shown later. As before, we will denote the setting for Alice and Bob with $x$ and $y$, respectively. For $x=1$, Alice will open Charlie's laboratory and ask him directly for his outcome $c$, assigning its value to her own outcome $a$. For other values of $x$, on the other hand, she will reverse the evolution of the laboratory to a state before Charlie's measurement, then proceed to perform her own measurement directly on the particle, with such measurement depending on the specific value of $x$. Correspondingly, Bob follows a similar procedure with regard to Debbie's laboratory for his measurement $y$. We will denote this scenario as $\mathcal{S}_{1,1}$. In general, we will define $\mathcal{S}_{k,l}$ as an extended Wigner's friend scenario constructed from a public scenario $\mathcal{S}$, where Alice makes $k$ measurements where she asks Charlie directly for his outcome, and likewise for $l$ with Bob and Debbie.

The main result of \cite{Bong20} can be expressed as follows:
\begin{thm} \label{thm:no:ew}
If a superobserver can perform arbitrary quantum operations on an observer and their environment, then there exists quantum phenomena for which there is no model satisfying local friendliness.
\end{thm}
This is quite similar to Bell's theorem (theorem \ref{thm:bell}), save for the assumption that there can be a superobserver capable of performing quantum operations on the friend, and that this theorem rules out the assumptions of local friendliness, as opposed to local determinism. Therefore, theorem \ref{thm:no:ew} drops the assumption of predetermination altogether. Theorem \ref{thm:no:ew} is also very similar to Brukner's theorem (theorem \ref{thm:brk}), in that the condition that a superobserver can perform arbitrary quantum operations on an observer and their environment corresponds to Brukner's assumption of universal validity of quantum theory, and his assumptions of locality and freedom of choice together imply local agency. However, the assumption of AOE is weaker than Brukner's assumption of observer-independent facts. In particular, it does not assume that there is a truth value for measurements that are not performed by any observer in a given run of an experiment.

In \cite{Bong20} theorem \ref{thm:no:ew} was proven by characterizing the polytope of correlations $\mathbb{LF}(\mathcal{S}_{1,1})$ that arises in this scenario for a given value of $M$ under the assumptions of Local Friendliness, and then testing, theoretically, if a quantum setup could violate the inequalities that describe the facets of said polytope, similarly to how the violation of a Bell inequality is produced by correlations outside $\mathbb{LD}(\mathcal{S})$. In fact, these new polytopes are all partially deterministic polytopes of the form $\mathbb{LF}(\mathcal{S}_{1,1}) = \mathbb{PD}_{\{1\},\{1\}}(\mathcal{S})$, which in turn implies that $\mathbb{LF}(\mathcal{S}_{1,1})\subseteq\mathbb{LD}(\mathcal{S})$. It was shown that, for $M=2$, it holds that $\mathbb{LF}(\mathcal{S}_{1,1}) = \mathbb{LD}(\mathcal{S})$, so the violation of a Bell inequality is sufficient to prove the theorem.

More interesting are the cases where $M>2$, as the polytope in question will be a proper superset of the local deterministic one for the same scenario. The case $M=3$ was studied in detail in \cite{Bong20}. It was seen there that the characterization of new categories of inequalites is necessary, as not all the facets of $\mathbb{LF}(\mathcal{S}_{1,1})$ are equivalent to Bell inequalities. Furthermore, the inverse is not true either, as not all Bell inequalites are facets of $\mathbb{LF}(\mathcal{S}_{1,1})$. Indeed, we have that CH inequalities are facets of $\mathbb{LF}(\mathcal{S}_{1,1})$ as long as they include at least one of the settings $x=1$ or $y=1$, but not otherwise. This last point suggests that we could consider a scenario where only one of Alice and Bob has a friend \cite{Wiseman22}, as we will do in the next section.

A proof-of-principle experiment for violations of facets of $\mathbb{LF}(\mathcal{S}_{1,1})$ was realized in \cite{Bong20}. Instead of an actual observer, however, a photon polarization qubit was used as the ``friend'', with the friend's ``observation'' corresponding to said photon's path. An experiment that considered an actual observer, for example, one that made use of a sufficiently advanced quantum computer, would need to address many more technological and ethical assumptions. These matters are discussed in detail in \cite{Wiseman22}.

\section{Sequential extended Wigner's friends scenarios}
\label{ch:sewfs}

We have seen that it is possible for quantum theory to violate the assumptions of Local Friendliness. We have also seen that, in general, $\mathbb{LF}(\mathcal{S}_{1,1})$ will contain $\mathbb{LD}(\mathcal{S})$, but that there are specific scenarios where these two are equal. Because we are much more familiar with the facets of $\mathbb{LD}(\mathcal{S})$ than those of $\mathbb{LF}(\mathcal{S}_{1,1})$, it would be of interest to find scenarios where these two sets are the same.

Let us consider three experimenters: Alice, Bob and Charlie. Alice and Charlie are in a Wigner's friend arrangement, with Alice as the superobserver and Charlie as the friend. Bob and Charlie share a two-particle system. Both parties are distant enough such that Alice's and Charlie's experiments are space-like separated from Bob's, as illustrated in Fig. \ref{fig:spacetime}. Bob's choice of measurement on his particle is labelled by $y$, with $N_B$ possible choices, and with output labelled by $b$.

\begin{figure}
    \centering
    \includegraphics[width=0.7\linewidth]{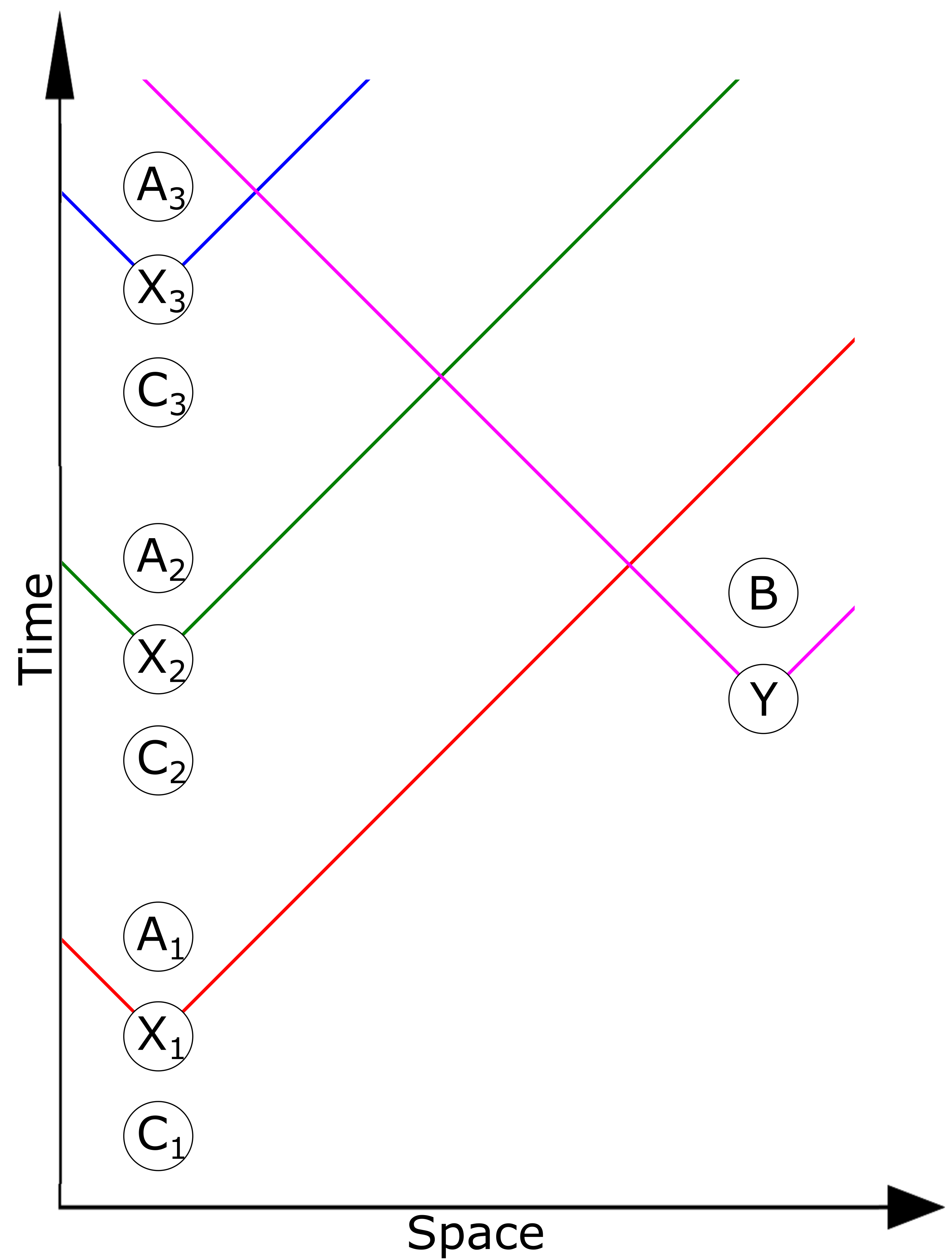}
    \caption{Space time diagram for relevant events in a sequential EWFS with $R=3$. Coloured lines represent the edges of the future light-cones for each of the free choices made in the experiment. Bob's choice of measurement setting $y$ is space-like separated from all events in the region where Alice and Charlie are, up until Alice observes her final possible outcome $a_3$. Similarly, Alice's choices for $x_1$, $x_2$ and $x_3$ are space-like separated from Bob's measurement events.}
    \label{fig:spacetime}
\end{figure}

Charlie, inside his laboratory, performs a measurement, labelled by $z_1$, on his particle. It should be remarked that $z_1$ is not a free choice, and is fixed over all iterations of the experiment. Alice, unlike in any of the previous scenarios, can choose how and when to measure Charlie's laboratory, out of a list of $R$ ways and times. The different potential measurements will be labelled with an index $i$, with $1\le i \le R$. At time $t=t_1$, Alice will have two possible choices. She can choose to open Charlie's laboratory and ask him directly for his outcome $c_1$ and register it as her outcome $a_1$, which corresponds to the input $x_1=1$. This will end Alice's experiment. If $x_1=0$ instead, then Alice will not open the laboratory. Rather, she will reverse its evolution so that its state is equal to what it had been at a point in time before Charlie's measurement, record $a_1=\varnothing$, and instructs Charlie to make measurement $z_2$, and the experiment continues. Again, $z_2$ is not a free choice made by any agent, and will be the same across all iterations of the experiment. As opposed to when Charlie communicates with Alice, when Alice instructs Charlie this does not change the quantum state of Charlie's laboratory, including the particle, as it is a closed system. (Note that this is just a result of quantum physics, not a metaphysical assumption.)

If $x_1=0$, then, at a later time $t=t_2$, Alice will make a second measurement, $i=2$, with a binary input $x_2$. As before, if $x_2=1$ Alice asks Charlie for his outcome $c_2$, which has no assumed relation to $c_1$, and her outcome will be $a_2=c_2$. If $x_2=0$ she will behave in a similar way as with $x_1=0$, with her reversing Charlie's laboratory to a state before his measurement was performed (but after she communicated with him and gave him new instructions), giving him instructions to perform measurement $z_3$, and setting $a_2=\varnothing$.

Successive $x_i$ follow the same formula.  Provided that she has not yet finished her experiment, Alice will either ask Charlie directly for the value of $c_i$ if $x_i=1$, so $a_i=c_1$, or reverse the evolution of Charlie's laboratory if $x_i=0$. In the latter case, Charlie will perform a new measurement, according to the latest instructions received, on his particle. Alice will register $a_i=\varnothing$ if $i\le R$, but she will measure the particle directly in a fixed basis if $i=R$, and so $a_R\neq \varnothing$. Any non-observed $c_i$ and $a_i$ will be considered to be equal to $\varnothing$. This protocol is illustrated in Fig. \ref{fig:swfs}.

\begin{figure}
    \centering
    \includegraphics[width=0.7\linewidth]{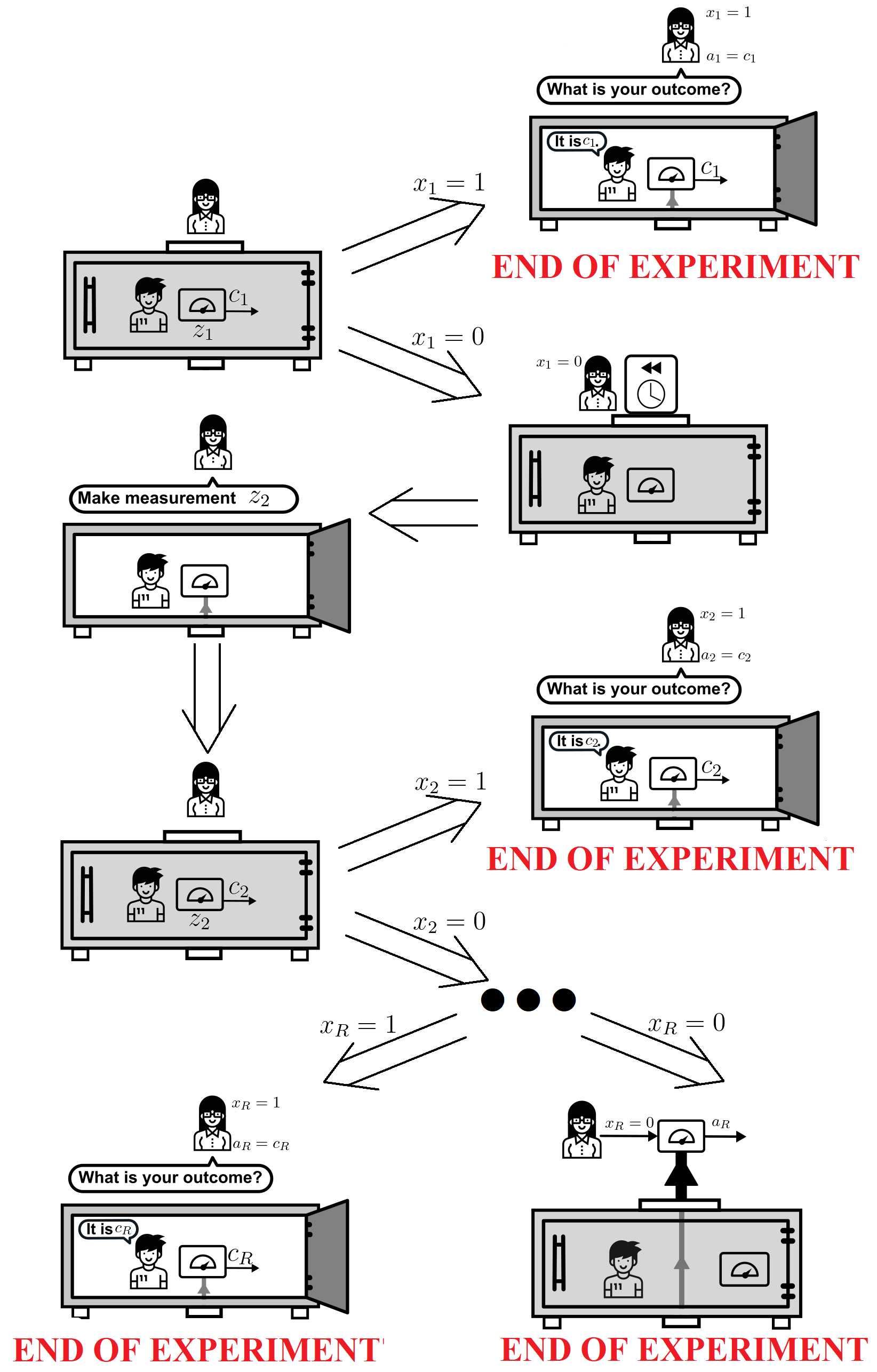}
    \caption{Protocol for Alice and Charlie in the sequential EWFS. Charlie performs a measurement $z_1$ on his outcome, observing outcome $c_1$. Then, either (if $x_1=1$) Alice opens the lab, asks Charlie for his outcome, records $a_1=c_1$ and ends the experiment, or (if $x_1=0$) Alice reverses Charlie's measurement and instructs him to perform a new measurment, $z_2$. This process repeats until Alice ends the experiment after asking Charlie directly for his outcome when $x_i=1$ in some round $i$, or, if all $x_i=0$ by round $i=R$, by measuring the particle herself, in some fixed basis, after reversing Charlie's last measurement.}
    \label{fig:swfs}
\end{figure}

For compactness, we will introduce the following notation. Let $\tilde{a}$ be the first non-null value among all the $a_i$, so it will correspond to the outcome observed by Alice at the end of the experiment. Let $\tilde{x}$ be the first among all the $i$ such that $x_i=1$, except in the case where $x_i=0$ for all $i$, in which case $\tilde{x}=R+1$. We also define $\LRc = (c_1,c_2,\cdots,c_R)$, and $\Lc{i} = (c_1,c_2,\cdots,c_i)$.

Within this scenario, the assumptions of local friendliness formally translate as
\begin{align}
    \mathbf{AOE:}\; &\exists\, p(\tilde{a}b\LRc|\tilde{x}y),\:\:\forall \tilde{a},b,\LRc,\tilde{x},y,\;\;\mathrm{s.t.}\nonumber\\
    & p(\tilde{a}|b,\LRc,\tilde{x}=i,y) = \delta_{\tilde{a},c_i},\;\;\forall\tilde{a},b,\LRc,y,1\le i \le R.\label{eq:sequential:macroreality} \\
    \nonumber \\
    \mathbf{Local\;agency:}\;& p(b\Lc{j}|x_i x_k y) = p(b\Lc{j}|x_k y),\;\;\forall b,i,y,k<j\le i, \label{eq:sequential:la:1}\\
    & p(\tilde{a}\LRc|\tilde{x}y) = p(\tilde{a}\LRc|\tilde{x}),\;\;\forall \tilde{a},\LRc,\tilde{x},y.\label{eq:sequential:la:2}
\end{align}

\section{Theorem}
\label{ch:theorem}

\begin{thm}
\label{thm:sewfs}
The violation of any Bell inequality within a sequential EWFS implies the violation of Local Friendliness.
\end{thm}

\begin{proof}

To prove the theorem, it is enough to show that the set of all correlations that obey Local Friendliness in a sequential EWFS is equal to the set of all local deterministic correlations in that same scenario. Firstly, we must characterize LF correlations within a sequential EWFS. For this, we define $\Rc{i}=(c_{i+1},\cdots,c_R)$ for $1\le i \le R-1$, and $\Rc{R}$ as the empty set.

\begin{align}
    p(\tilde{a}b|\tilde{x}=i,y) &= \sum_{\LRc} p(\tilde{a}b\LRc|\tilde{x}=i,y)\\
    &= \begin{cases}
 \sum_{\LRc} p(\tilde{a}|b\LRc,\tilde{x}=i,y) p(b \Rc{i}|\Lc{i},\tilde{x}=i,y) p(\Lc{i}|\tilde{x}=i,y) & \text{ if } 1\le i \le R \\ 
 \sum_{\LRc} p(\tilde{a}b|\LRc,\tilde{x}=i,y) p(\LRc|\tilde{x}=i,y) & \text{ if } i=R+1.
\end{cases}
\end{align}

From (\ref{eq:sequential:macroreality}), for $1\le i \le R$,
\begin{equation}
     \sum_{\LRc} p(\tilde{a}|b\LRc,\tilde{x}=i,y) p(b \Rc{i}|\Lc{i},\tilde{x}=i,y) p(\Lc{i}|\tilde{x}=i,y) = \sum_{\LRc} \delta_{\tilde{a},c_i} p(b \Rc{i}|\Lc{i},\tilde{x}=i,y) p(\Lc{i}|\tilde{x}=i,y)
\end{equation}


From (\ref{eq:sequential:la:2})
\begin{align}
    p(\tilde{a}b|\tilde{x}=i,y) &= \begin{cases}
 \sum_{\LRc} \delta_{\tilde{a},c_i} p(b \Rc{i}| \Lc{i},\tilde{x}=i,y) p(\Lc{i}|\tilde{x}=i)& \text{ if } 1\le i \le R, \\ 
 \sum_{\LRc} p(\tilde{a}b|\LRc,\tilde{x}=i,y) p(\LRc|\tilde{x}=i) & \text{ if } i=R+1.
\end{cases}\\
    &=  \begin{cases}
    \sum_{\Lc{i}} \delta_{\tilde{a},c_i} p(b|\Lc{i},\tilde{x}=i,y) p(\Lc{i}|\tilde{x}=i)& \text{ if } 1\le i \le R, \\ 
     \sum_{\LRc} p(\tilde{a}b|\LRc,\tilde{x}=i,y) p(\LRc|\tilde{x}=i) & \text{ if } i=R+1.
    \end{cases} \label{eq:sequential:step:1}
\end{align}

From (\ref{eq:sequential:la:1}),
\begin{equation}
    p(\Lc{j}|x_i x_k) = p(\Lc{j}|x_k),\;\;\forall \LRc, k < j \le i.
\end{equation}
Then,
\begin{equation}\label{eq:sequential:pre:pcj}
    p(\Lc{i}|\tilde{x}=i) = p(\Lc{i}|x_1 = \cdots = x_{i-1} = 0, x_{i} \cdots x_R), \;\;\forall x_i,\cdots, x_R.
\end{equation}
That is to say, the probability of $\Lc{i}$ conditioned on $\tilde{x}$ will be the same among all values of $\tilde{x}\ge i$, so
\begin{equation}\label{eq:sequential_pcj}
     p(\Lc{i}|\tilde{x}=i) = p(\Lc{i}|\tilde{x}=R+1).
\end{equation}

Also from (\ref{eq:sequential:la:1}):
\begin{equation}
    p(b|\Lc{j} x_i x_k y) = p(b | \Lc{j} x_k y),\;\;\forall \LRc, k, y, < j \le i.
\end{equation}
Then, in a similar fashion to Eqs. (\ref{eq:sequential:pre:pcj}) and (\ref{eq:sequential_pcj}), we have
\begin{align}
    p(b|\Lc{i}, \tilde{x}=i, y) &= p(b |\Lc{i}, x_1 = \cdots = x_{i-1} = 0, x_{i} \cdots x_R, y), \;\;\forall x_i,\cdots, x_R\\
    &= p(b | \Lc{i}, \tilde{x}=R+1, y).\label{eq:sequential_pb}
\end{align}
We define $\hat{p}$ such that $\hat{p}(\Lc{i}) = p(\Lc{i}|\tilde{x}=R+1)$ and $\hat{p}(b |\Lc{i} y) = p(b | \Lc{i}, \tilde{x}=R+1, y)$. Thus from (\ref{eq:sequential:step:1}), \eqref{eq:sequential_pcj} and \eqref{eq:sequential_pb},
\begin{equation}
\label{eq:sequential:step:2}
    p(\tilde{a}b|\tilde{x}=i,y) = \begin{cases}
    \sum_{\Lc{i}} \delta_{\tilde{a},c_i} \hat{p}(b|\Lc{i} y) \hat{p}(\Lc{i})& \text{ if } 1\le i \le R, \\ 
     \sum_{\LRc} p(\tilde{a}b|\LRc,\tilde{x}=i,y) \hat{p}(\LRc) & \text{ if } i=R+1.
    \end{cases}
\end{equation}

We need all terms to depend on a distribution of the same ``hidden variable''. For that, we note
\begin{align}
    \sum_{\Lc{i}} \delta_{\tilde{a},c_i} \hat{p}(b|\Lc{i} y) \hat{p}(\Lc{i}) &= \sum_{\Lc{i},c_{i+1},\cdots,c_R} \delta_{\tilde{a},c_i} \hat{p}(b|\Lc{i} y) \hat{p}(\Lc{i} c_{i+1} \cdots c_R)\\
    &= \sum_{\LRc} \delta_{\tilde{a},c_i} \hat{p}(b|\Lc{i} y) \hat{p}(\LRc).
\end{align}
Replacing in (\ref{eq:sequential:step:2}),
\begin{equation}
\label{eq:sequential:step:3}
    p(\tilde{a}b|\tilde{x}=i,y) = \begin{cases}
    \sum_{\LRc} \delta_{\tilde{a},c_i} \hat{p}(b|\Lc{i} y) \hat{p}(\LRc)& \text{ if } 1\le i \le R, \\ 
     \sum_{\LRc} p(\tilde{a}b|\LRc,\tilde{x}=i,y) \hat{p}(\LRc) & \text{ if } i=R+1.
    \end{cases}
\end{equation}
The set of all correlations that satisfy (\ref{eq:sequential:step:3}) is denoted by $\mathbb{SW}(\mathcal{S}_{R,0})$.

In (\ref{eq:sequential:step:3}) the term $p(\tilde{a}b|\LRc,\tilde{x}=R+1,y)$ is only constrained by the no-signalling principle and positivity requirements, so it can be any behaviour in $\mathbb{NS}((\mathcal{A,B,}\{R+1\},\mathcal{Y}))$. The extreme points of $\mathbb{SW}(\mathcal{S}_{R,0})$ can be characterized with a variable $\zeta$ that determines the values of $\LRc(\zeta)$ and $j(\zeta)$, with $j(\zeta)$ labelling the extreme points of $\mathbb{NS}((\mathcal{A,B,}\{R+1\},\mathcal{Y}))$. Therefore, (\ref{eq:sequential:step:3}) can be rewritten as
\begin{equation}
\label{eq:sequential:step:4}
    p(\tilde{a}b|\tilde{x}=i,y) = \begin{cases}
    \sum_{\zeta} \delta_{\tilde{a},c_i(\zeta)} p^{j(\zeta)}_{\mathrm{Ext}}(b|y) \hat{p}(\zeta)& \text{ if } 1\le i \le R, \\ 
     \sum_{\zeta}  p^{j(\zeta)}_{\mathrm{Ext}}(\tilde{a}b|\tilde{x}=i,y) \hat{p}(\zeta) & \text{ if } i=R+1,
    \end{cases}
\end{equation}
where $p^{j(\zeta)}_{\mathrm{Ext}}(\tilde{a}b|\tilde{x}=i,y)$ is the value of $p(\tilde{a}b|\tilde{x}=i,y)$ at the extreme point of $\mathbb{NS}((\mathcal{A,B,}\{R+1\},\mathcal{Y}))$ labelled by $j(\zeta)$. Then, the number of extreme points in $\mathbb{SW}(\mathcal{S}_{R,0})$ will be $\#(\{j(\zeta)\}) \times \#(\{c(\zeta)\})$. If $c(\zeta)$ has a finite number of possible values, then $\mathbb{SW}(\mathcal{S}_{R,0})$ will have a finite number of extreme points.

Let us consider two correlations that satisfy (\ref{eq:sequential:step:4}), $\bar{p}_1 (\tilde{a}b|\tilde{x}y) = (p_1 (\tilde{a}b|\tilde{x},y))$ and $\bar{p}_2 (\tilde{a}b|\tilde{x}y) = (p_2 (\tilde{a}b|\tilde{x},y))$. We define $\bar{p}' = \alpha \bar{p}_1 + (1-\alpha) \bar{p}_2$, with $0<\alpha<1$. Then,
\begin{equation}
\label{eq:sequential:convex}
    p'(\tilde{a}b|\tilde{x}=i,y) = \begin{cases}
    \sum_{\zeta} \delta_{\tilde{a},c_i(\zeta)} p^{j(\zeta)}_{\mathrm{Ext}}(b|y) (\alpha \hat{p}_1(\zeta) + (1-\alpha) \hat{p}_2(\zeta)) & \text{ if } 1\le i \le R, \\ 
     \sum_{\zeta}  p^{j(\zeta)}_{\mathrm{Ext}}(\tilde{a}b|\tilde{x}=i,y) (\alpha \hat{p}_1(\zeta) + (1-\alpha) \hat{p}_2(\zeta)) & \text{ if } i=R+1,
    \end{cases}
\end{equation}
By defining $\hat{p}'(\zeta) = \alpha\hat{p}_1(\zeta) + (1-\alpha)\hat{p}_2(\zeta)$, we have
\begin{equation}
    p'(\tilde{a}b|\tilde{x}=i,y) = \begin{cases}
    \sum_{\zeta} \delta_{\tilde{a},c_i(\zeta)} p^{j(\zeta)}_{\mathrm{Ext}}(b|y) \hat{p}'(\zeta)& \text{ if } 1\le i \le R, \\ 
     \sum_{\zeta}  p^{j(\zeta)}_{\mathrm{Ext}}(\tilde{a}b|\tilde{x}=i,y) \hat{p}'(\zeta) & \text{ if } i=R+1,
    \end{cases}
\end{equation}
This is in the same form as (\ref{eq:sequential:step:4}), so it follows that $\bar{p}'$ is in $\mathbb{SW}(\mathcal{S}_{R,0})$. Therefore, $\mathbb{SW}(\mathcal{S}_{R,0})$ is a convex set. Because it has a finite number of extreme points, it is in fact a convex polytope. 

For $1\le i \le R$, we have that $p(\tilde{a}| \tilde{x}=i, y) = \sum_{b} p(\tilde{a}b| \tilde{x}=i, y) = \sum_{b,\zeta} \delta_{\tilde{a},c_i(\zeta)} p^{j(\zeta)}_{\mathrm{Ext}}(b|y) \hat{p}'(\zeta) = \sum_{\zeta} \delta_{\tilde{a},c_i(\zeta)} \hat{p}'(\zeta)$. At extreme points, the value of 
$\LRc$ is fixed, and so $p(\tilde{a}| \tilde{x}\neq R+1, y,\zeta) \in \{0,1\}$ at those points. That is to say, such points are deterministic for $i\in\mathcal{X}$ with $i\neq R+1$. Thus,
\begin{equation}
    \mathbb{SW}(\mathcal{S}_{R,0}) = \mathbb{PD}_{\mathcal{X}/{R+1},\emptyset}(\mathcal{S})
\end{equation}

From (\ref{eq:pd:to:ld}), we have finally
\begin{equation}
    \mathbb{SW}(\mathcal{S}_{R,0}) = \mathbb{LD}(\mathcal{S})
\end{equation}

\end{proof}

\section{Discussion}
\label{ch:discussion}

This paper had the purpose of addressing two questions, originally stated in the introduction. Is the gap between Bell inequality violation and LF inequality violation a fundamental property of the LF assumptions themselves, or is it a limitation of the scenarios to which they are being applied? What other scenarios can be constructed that include the reversal of the friend's measurement, and what can they be used for? With the sequential EWFS we address the second question, by presenting a scenario that can include several reversals of the friend's measurement. For the first question, theorem \ref{thm:sewfs} shows that the gap between violations of the two types of inequalities is dependent on the scenario under consideration, and it is possible to construct a scenario, namely, the sequential EWFS, where there is never a gap at all.

Note that this result is not in contradiction with the fact that the LF assumptions are metaphysically weaker than Bell’s. It just emphasises that the distinction between those assumptions is only expressed in some scenarios. In a standard Bell scenario without any friends, for example, the LF assumptions do not lead to inequalities that can be violated by quantum correlations.

There is no reason to believe that the sequential EWFS are the only scenarios that can close the gap between Bell and LF inequalities. It would be of interest for future work to find other EWFS that also exhibit this property. For example, one possible variation of the EWFS that may yield interesting results would be to consider a ``nested'' sort of arrangement, where a friend can have a friend of their own, and this friend can also have a friend, and so on. Experimental realizations for the sequential EWFS, or other EWFS of interest, such as the nested one just described, could be implemented, similarly to the experiments performed in \cite{Bong20, Proietti19}. Finally, we note that that the sequential scheme EWFS of this paper, originally presented in the first author's thesis~\cite{Utreras22}, has also recently been applied to another theorem related to Wigner's friend~\cite{Omrod23} which uses somewhat different assumptions from Local Friendliness.

\textbf{Funding.} This work was supported by grant number FQXi-RFP-CPW-2019 from the Foundational Questions Institute and Fetzer Franklin Fund, a donor advised fund of Silicon Valley Community Foundation, and by the Centre for Quantum Computation and Communication Technology (CQC2T). A.U.-A. acknowledges financial support through Australian Government Research Training Program Scholarships.

\textbf{Acknowledgements.} We acknowledge the traditional owners of the land at Griffith University on which this work was undertaken, the Yuggera and Yugambeh peoples. Avatars in Figs. \ref{fig:ewfs} and \ref{fig:swfs} are adapted from Eucalyp Studio, available under a Creative Commons licence (Attribution 3.0 Unported), \url{https://creativecommons.org/licenses/by/3.0/}, at \url{https://www.iconfinder.com/iconsets/avatar-55}.

\end{document}